\documentclass[copyright,creativecommons]{eptcs}
\providecommand{\event}{CL&C 2010}
\providecommand{\urlalt}[2]{\href{#1}{#2}} \providecommand{\doi}[1]{doi:\urlalt{http://dx.doi.org/#1}{#1}} 
\usepackage{amsmath}

\usepackage{amssymb}
\usepackage{mathrsfs}
\usepackage{bussproofs}
\usepackage{upgreek}
\usepackage{microtype}
\usepackage{hyperref}

\usepackage{pdfsync}

\usepackage{mathrsfs}
\usepackage{a4wide}

\newcommand{\Long}  [2]                {#2}
\newcommand{\PaperState}               {\Long}
\newcommand{\IfPaperState}  [2]        {\PaperState{#1}{#2}}

\newcommand{\EM}                       { {\mathsf{EM}} }
\newcommand{\SK}                       { {\mathsf{SK}} }
\newcommand{\AC}                       { {\mathsf{AC}} }
\newcommand{\HA}                       { {\mathsf{HA}} }

\newcommand{\PA}                       { {\mathsf{PA}} }

\newcommand{\Nat}                      { {\tt N} }
\newcommand{\Bool}                     { {\tt Bool} }
\newcommand{\State}                    { {\tt S} }
\newcommand{\Update}                    { {\tt U} }
\newcommand{\NatSet}                   {\mathbb{N}}

\newcommand{\UpdateSet}                 {\mathbb{U}}
\newcommand{\SystemT}                  {\mathcal{T}}
\newcommand{\SystemTG}                  {\mathsf{T}}

\newcommand{\True}                     { {\tt{True}} }
\newcommand{\False}                    { {\tt{False}} }

\newcommand{\Class}                    {\mbox{\tiny Class}}

\newcommand{\SystemTClass}             {{\SystemT_{\Class}}}

\newcommand{\Language}                 {\mathcal{L}}
\newcommand{\LanguageClass}            {\Language_{\Class}}

\newcommand{\comment}[1]{}
\newcommand{\proj}                     { {\mathsf{p}} }
\newcommand{\CupSem}                   { {\mathcal U} }

\newcommand{\funn}                   { {\Nat\rightarrow\Nat}}
\newcommand{\ifthen} [3]                        { {\mathsf{if}\ {#1}\ \mathsf{then}\ {#2}\ \mathsf{else}\ {#3} } }

\newcommand{\rec}                          {{\mathsf{R}}}
\newcommand{\ifn}                          {{\mathsf{if}}}
\newcommand{\suc}                      {{\mathsf{S}}}

\newcommand{\Phic}                        {{\mathsf{\Phi}}}

\newcommand{\upconst}[1]           {{\overline{#1}}}

\newcommand{\truth}[1]             {{\upchi}_{#1}}
\newcommand{\truthst}[2]             {{{#1}^{#2}}}
\newcommand{\truthstbot}[2]             {{({#1}^{\bot})^{#2}}}

\newcommand{\substitution} [1]         { {\overline{#1}} }

\newcommand{\dom}[2]          {{\mathsf{dom}_{#1}(#2)}}

\newcommand{\minu}                {{\mathsf{min}}}
\newcommand{\get}             {{\mathsf{get}}}

\newcommand{\mkupd}            {{\mathsf{mkupd}}}

\newcommand{\tkleene}   {{\mathscr{T}}}

\newcommand{\mrs}             {{\mathsf{mr}}}
\newcommand{\mrsf}             {{\Vdash}}

\newcommand{\empredicate}[1]    {{\mathsf{P}_{#1}}}

\newcommand{\skolemcon}  {{\mathcal{SC}}}
\newcommand{\code}[1]       {{\langle #1\rangle}}
\newcommand{\lev}[1]         {{\mathsf{lev}(#1)}}

\newcommand{\real}[1]       {{\boldsymbol{{\mathfrak{T}}}_{#1}}}
\newcommand{\freal}[1]       {{\boldsymbol{\mathfrak{F}}_{#1}}}

\newcommand{\sub}             {{\,\smallsetminus\,}}

\newcommand{\dm}[1]{\mathsf{def}(#1)}
\newcommand{\rs}[1]{\overset {#1}{\mapsto}}
\newcommand{\comp}[1]{\| #1 \|}
\newcommand{\redto}{\rightarrow}
\newcommand{\dlinea}{\leavevmode\hrule\vspace{1pt}\hrule\mbox{}}
\newcommand{\notreal}{\Vvdash\hspace{-2.8ex}\scriptsize \diagup}
\newcommand{\sbs}[1]{\overline{#1}}
    \newtheorem{notation}{Notation}
\newcommand{\tdef}[1]{\downarrow^{\hspace{-0.3ex}#1}}
\newcommand{\FSet}{\Upgamma}
\newcommand{\HAw}{\HA^{\omega}}

\newtheorem{theorem}{Theorem}

\newtheorem{lemma}{Lemma}
\newtheorem{proposition}{Proposition}
\newtheorem{definition}{Definition}

\title{Interactive Realizability and the elimination of Skolem functions in Peano Arithmetic}

\author{Federico Aschieri
\institute{Laboratoire PPS, \'equipe PI.R2, \\
Universit\'e Paris 7, INRIA$\And$CNRS}
\and
Margherita Zorzi\thanks{Supported by ANR
COMPLICE project (Implicit Computational Complexity, Concurrency and Extraction), ref.: ANR-08-BLANC-0211-01.}
\institute{Laboratoire d'Informatique de Paris-Nord \\
UMR CNRS 7030 \\
Institut Galil\'ee -- Universit\'e Paris-Nord}
}
\def\titlerunning{Interactive Realizability and the elimination of Skolem functions in Peano Arithmetic}
\def\authorrunning{F. Aschieri \& M. Zorzi}

\begin{document}

\maketitle

\begin{abstract}

We present a new syntactical proof that first-order Peano Arithmetic with Skolem axioms is conservative over Peano Arithmetic alone for arithmetical formulas.
This result -- which shows that the Excluded Middle principle can be used to eliminate Skolem functions -- has been previously proved by other techniques, among them the epsilon substitution method and forcing.  In our proof, we employ  Interactive Realizability, a computational semantics for Peano Arithmetic which extends Kreisel's modified realizability to the classical case.
\end{abstract}


\section{Introduction}

\vspace{-1.5ex}

For a long time it has been known that intuitionistic realizability can be used as a flexible  tool for obtaining a wealth of unprovability, conservativity and proof-theoretic results \cite{TroelstraRealizability,TroelstraHandBook}. As title of example, with Kreisel's modified realizability \cite{Kreiselm}, one can show the unprovability of Markov Principle in Heyting Arithmetic in all finite types ($\HAw$) and the conservativity of $\HAw$ with the Axiom of Choice ($\AC$) over $\HAw$ for negative formulas.
In both cases, one starts by showing that any formula provable in one of those systems can be shown to be realizable in $\HAw$.
In the first case, one proves that the realizability of Markov Principle implies the solvability of the Halting Problem, and concludes that Markov Principle is unprovable in $\HAw$. In the second, one exploits the fact that the assertion ``$t$ realizes $A$'' is exactly the formula $A$ when $A$ is  negative and concludes that $\HAw$ proves $A$.

The situation in classical logic has been very different: for a long time it did not exist any realizability notion suitable to interpret directly classical proofs, let alone proving independence or conservation results. However, recently several classical realizability interpretations have been put forward.  Among them: Krivine's classical realizability~\cite{Krivine}, which has been shown in~\cite{Krivine2} to yield striking unprovability results in Zermelo-Fraenkel set theory,
and Interactive realizability~\cite{AB,ABF,AschieriHAS,AschieriPA}, which has been shown in \cite{Aschierilearning,AschieriHAS} to provide conservation results for $\Pi_{2}^{0}$-formulas.

Being a tool for  extracting programs from proofs, it is however quite natural that Interactive realizability is capable of producing $\Pi_{2}^{0}$-conservativity results.
The aim of this paper is to prove  that Interactive realizability can as well be used to prove other conservativity results.
In particular, let us consider first-order classical  Peano Arithmetic $\PA$, which is $\HA+\EM$, where $\EM$ is the excluded middle over arithmetical formulas. Then we give a new syntactic proof that $\PA$ with the Skolem axiom scheme $\SK$ is conservative over $\PA$ for arithmetical formulas -- a result first syntactically proven by Hilbert and Bernays \cite{HilbertBernays} by means of the epsilon substitution method. The result is particularly interesting since it implies that classical choice principles can be eliminated by using the excluded middle alone. The structure of our proof resembles the pattern of the intuitionistic-realizability conservation proofs we have sketched above and allows to obtain a stronger result. Namely, we shall show that if an arithmetical formula $A$ is provable in $\HAw+\EM+\SK$, then the assertion ``$t$ realizes $A$'' is provable in $\HAw$ alone.  Afterwards, we shall show the provability in $\HAw+\EM$ of the assertion ``($t$ realizes $A$) implies $A$'' and thus conclude that $\HAw+\EM$ proves $A$. Since this latter system is conservative over $\PA$ for arithmetical formulas, we obtain the result.

In our opinion, there are at least two reasons our proof technique is interesting. As remarked by Avigad\cite{AvigadSkolem}, the methods based on the epsilon-method, Herbrand's Theorem or cut-elimination lead to an exponential increase in the size of the proof, when passing from a proof in $\HAw+\EM+\SK$ to a corresponding proof in $\HA+\EM$; instead, we conjecture that our transformation is polynomial. To the best of our knowledge, there is only another method that does equally well, which is Avigad's \cite{AvigadSkolem}. The technique of Avigad is related to ours since it uses the method of forcing, in which the conditions are finite approximations of the Skolem functions used in the proof. With forcing one avoids speaking about infinite non-computable objects (i.e. the Skolem functions) and can approximate the original proof. Avigad's method is very simple and elegant when there is only one Skolem function to eliminate, but it becomes more complicated and difficult to handle when dealing with several Skolem functions. In fact, a nesting of the notion of forcing together with a technical result about elimination of definitions become necessary and the method loses some intuitive appeal. Instead, the use of Interactive realizability allows to deal with all the Skolem functions at the same time, and we conjecture that the resulting proofs are much shorter than the ones obtained by forcing. Moreover, the notion of forcing as an approximation of model-theoretic truth is harder to come up with, and it is much more natural to talk about states and approximations when dealing with programs.

Secondly, the theory of Interactive realizability offers a uniform explanation of a number of different phenomena. Rather than proving each particular meta-theoretic result about classical Arithmetic with an ad-hoc technique, one employs a single methodology. For example, one may prove conservativity of $\PA$ over $\HA$ for $\Pi^{0}_{2}$-formulas by a negative translation followed by Friedman's translation \cite{Friedman}; one may extract from proofs terms of G\"odel's System $\SystemTG$ by realizability or functional interpretations \cite{Godeldialectica}; one may prove the result about the elimination of Skolem functions with forcing; one may extract from proofs strategies in backtracking Tarski games by analyzing sequent calculus proofs \cite{Coquand}; one may obtain a simple ordinal analysis of $\PA+\SK$ by using update procedures \cite{Avigad}. Instead, with the theory of Interactive realizability one obtains all the results above as a consequence of a single concept (see \cite{Aschierilearning,Aschierigames,AschieriPA}).

\vspace{-3ex}

\paragraph{Plan of the paper}
In Section \S \ref{section-TheTermCalculus} we review the term calculus $\SystemTClass$ in which Interactive realizers are written, namely an extension of G\"odel's system $\SystemTG$ plus Skolem function symbols for a countable collection of Skolem functions.
In Section \S \ref{section-ALearningBasedRealizability} we recall Interactive realizability, as described in \cite{AschieriPA}, a computational semantics for  $\HA^{\omega}+\EM+\SK$, an arithmetical system with functional variables which includes first-order classical Peano Arithmetic and Skolem axioms.
 In Section \S \ref{section-Conservativity} we use Interactive realizability to prove the conservativity of $\HAw+\EM+\SK$ over $\HAw+\EM$ for arithmetical formulas.
In Section \S \ref{section-formalization} we explain in more detail how to formalize the proofs of Section~\ref{section-Conservativity} in $\HAw+\EM$ and $\HA+\EM$.

\section{The Term Calculus $\SystemTClass$}\label{section-TheTermCalculus}

\vspace{-1.5ex}

In this section we follow \cite{AschieriPA} and recall the typed lambda calculi $\SystemT$ and $\SystemTClass$ in which interactive realizers are written. $\SystemT$ is an extension  of G\"odel's system $\SystemTG$ (as presented in Girard \cite{Girard}) with some syntactic sugar. The basic objects of $\SystemT$ are numerals, booleans, and its basic computational constructs are primitive recursion at all types, if-then-else, pairs, as in G\"odel's $\SystemTG$. $\SystemT$ also includes as basic objects  finite partial functions over $\NatSet$ and simple primitive recursive operations over them. $\SystemTClass$ is obtained from $\SystemT$ by adding on top of it a collection of Skolem function symbols $\Phic_{0}, \Phic_{1}, \Phic_{2}, \ldots,$ of type $\Nat\rightarrow \Nat$, one for each arithmetical formula. The symbols are  inert from the computational point of view and realizers are always computed with respect to some approximation of the Skolem maps represented by $
\Phic_{0}, \Phic_{1},\Phic_{2},\ldots$.

\subsection{Updates}

\vspace{-1.5ex}

In order to define $\SystemT$, we start by introducing the concept of ``update'', which is nothing but a finite partial function over $\NatSet$. Realizers of atomic formulas will return these finite partial functions, or ``updates'', as new pieces of information that they have learned about the Skolem function $\Phic_{0}, \Phic_{1},\ldots$. Skolem functions, in turn, are used as ``oracles'' during computations in the system $\SystemTClass$. Updates are new associations input-output that are intended to correct, and in this sense, to \emph{update}, wrong oracle values used in a computation.

\begin{definition}
[Updates and Consistent Union]
\label{definition-StateOfKnowledge}
We define:
\begin{enumerate}

\comment{item
A binary {\em predicate} of $\SystemTG$ is any closed normal term $P:\Nat^{2}\rightarrow \Bool$ of G\"odel's $\SystemTG$.

\item
We assume $\empredicate{0}, \empredicate{1}, \empredicate{2}, \ldots $ is an arbitrary enumeration of all binary predicates of $\SystemTG$.\\}

\item
An update set $U$, shortly an \emph{update}, is a finite set of triples of natural numbers representing a finite partial function from $\NatSet^{2}$ to $\NatSet$. \\

\vspace{-2ex}
\item
Two triples $(a,n,m)$ and $(a',n', m')$ of numbers are {\em consistent} if  ${a} = {a}'$ and $n=n'$ implies $m = m'$.
Two updates $U_1, U_2$ are consistent if $U_1 \cup U_2$ is an update.\\

\vspace{-2ex}
\item
$\UpdateSet$ is the set of all updates.\\

\vspace{-2ex}

\item
The {\em consistent union} $U_1\, \CupSem\, U_2$ of $U_1, U_2 \in \UpdateSet$ is $U_1 \cup U_2$ minus all triples of $U_2$ which are inconsistent with some triple of $U_1$.

\end{enumerate}
\end{definition}

The consistent union $U_1\, \CupSem\, U_2$ is an non-commutative operation: whenever a triple of $U_1$ and a triple of $U_2$ are inconsistent, we arbitrarily keep the triple of $U_1$ and we reject the triple of $U_2$, therefore for some $U_1, U_2$ we have $U_1\, \CupSem\, U_2 \not = U_2\, \CupSem\, U_1$. $\CupSem$ represents a way of selecting a consistent subset of $U_1 \cup U_2$, such that  $U_1 \CupSem U_2 = \emptyset \implies U_1 = U_2 = \emptyset$.

\subsection{The System $\SystemT$}

\vspace{-1.5ex}

 $\SystemT$  is formally described in figure \ref{fig:F}. Terms of the form $\ifn_{A}\,t_1\,t_2\,t_3$ will be sometimes written in the more legible form $\ifthen{t_1}{t_2}{t_3}$. A \emph{numeral} is a term of the form $\suc(\ldots \suc(0)\ldots)$. For every update $U\in\UpdateSet$, there is in $\SystemT$ a constant $\upconst{U}: \Update$, where $\Update$ is a new base type representing $\UpdateSet$. We write $\varnothing$ for $\upconst{\emptyset}$. In $\SystemT$, there are four operations involving updates (see figure \ref{fig:F}):

\begin{enumerate}

\item
The first operation is denoted by the constant $\minu: \Update\rightarrow \Nat$. $\minu$ takes as argument an update constant $\upconst{U}$; it returns the minimum numeral $a$ such that $(a,n,m) \in U$ for some $n, m \in \Nat$, if any exists; it returns $0$ otherwise.\\

\vspace{-2.5ex}

\item
The second operation is denoted by the constant $\get: \Update\rightarrow \Nat^{3}\rightarrow \Nat$. $\get$ takes as arguments an update constant $\upconst{U}$ and three numerals $a,n, l$; it returns $m$ if $(a,n,m) \in U$ for some $m \in \Nat$ (i.e. if $(a,n)$ belongs to the domain of the partial function $U$); it returns $l$ otherwise.\\

\vspace{-2.5ex}

\item
The third operation is denoted by the constant $\mkupd: \Nat^{3}\rightarrow \Update$. $\mkupd$ takes as arguments three numerals $a,n,m$ and transforms them into (the constant coding in $\SystemT$) the update $\{(a,n,m)\}$.\\

\vspace{-2.5ex}

\item The forth operation is denoted by the constant $\Cup: \Update^{2}\rightarrow \Update$. $\Cup$ takes as arguments two update constants and returns the update constant denoting their consistent union.

\end{enumerate}

We observe that the constants $\minu, \get, \mkupd, \upconst{U}, \Cup$, and the type $\Update$  are just syntactic sugar and may be avoided by coding finite partial functions into natural numbers. System $\SystemT$ may thus be coded in G\"odel's $\SystemTG$.
\begin{figure*}[!htb]

\vspace{-1ex}
\dlinea

\vspace{-2ex}

\footnotesize{
\begin{description}
\item[Types] \[\sigma, \tau ::=  \Nat\ |\ \Bool\ |\ \Update\ |\ \sigma\rightarrow \tau\ |\ \sigma\times \tau\  \] \item[Constants]\[c::=\rec_{\tau} \ |\ \ifn_{\tau} \ |\ 0\ |\ \suc\ |\ \True\ |\ \False \ | \ \minu\ | \ \get\ |\ \mkupd\ | \ \Cup\ |\ \upconst{U} \text{ ($\forall U\in \UpdateSet$)}\]
\item[Terms]\[t,u::=\ c\ |\ x^\tau\ |\ tu\ |\ \lambda x^\tau u\  |\ \langle t, u\rangle\ |\ \pi_0u\ |\ \pi_{1}u\]
\item[Typing Rules for Variables and Constants]\[
x^\tau:\ \tau\ |\ 0:\ \Nat\ |\
\suc:\ \funn\ |\ \True:\ \Bool\ |\
\False:\ \Bool\ |\
\upconst{U}:\ \Update \text{ (for every $U\in \UpdateSet$)}\ |\
\Cup:\ \Update\rightarrow \Update\rightarrow \Update\]
\[|\ \minu:\ \Update\rightarrow \Nat\ |\
\get:\ \Update\rightarrow\funn\rightarrow \funn\ |\
\mkupd:\ \Nat\rightarrow \funn \rightarrow \Update\]
\[
|\ \ifn_{\tau}:\ \Bool\rightarrow \tau\rightarrow  \tau\rightarrow \tau\ |\
\rec_{\tau}:\  \tau \rightarrow (\Nat \rightarrow (\tau \rightarrow \tau))\rightarrow \Nat\rightarrow \tau
\]
\item[Typing Rules for Composed Terms]
\[\AxiomC{$t: \sigma\rightarrow \tau$}
\AxiomC{$u: \sigma$}
\BinaryInfC{$tu: \tau$}
\DisplayProof\qquad
\AxiomC{$u: \tau$}
\UnaryInfC{$\lambda x^\sigma u: \sigma\rightarrow \tau$}
\DisplayProof \qquad
\AxiomC{$u: \sigma$}
\AxiomC{$t: \tau$}
\BinaryInfC{$\langle u, t\rangle: \sigma\times \tau$}
\DisplayProof\qquad
\AxiomC{$u: \tau_{0}\times \tau_{1}$}
\RightLabel{$i\in\{0,1\}$}
\UnaryInfC{$\pi_iu: \tau_i$}
\DisplayProof\]

\item[Reduction Rules] All the usual reduction rules for simply typed lambda calculus (see Girard \cite{Girard}) plus the rules for recursion, if-then-else and projections
\[\rec_{\tau}uv0\mapsto u\quad \rec_{\tau}uv\suc(t)\mapsto vt(\rec_{\tau}uvt)\quad
\ifn_{\tau}\,\True\, u\,v\mapsto u\quad \ifn_{\tau}\, \False\, u\, v\mapsto v\quad \pi_i\langle u_0, u_1\rangle\mapsto u_i, i=0,1\]
plus the following ones, assuming $a,n,m, l$ be numerals:
\[\begin{aligned}\minu\, \upconst{U}
 &\mapsto
\begin{cases}a &\mathsf{if}\ \exists m,n.\ (a,{n},{m}) \in U\land \forall(b,i,j)\in U.\ a\leq b \\
0 &\mathsf{otherwise}
\end{cases}&\qquad \upconst{U}_1\Cup \upconst{U}_2&\mapsto\ \upconst{U_1\,\CupSem\, U_2}\\
\get\, \upconst{U}\, a\,{n}\,l &\mapsto
\begin{cases}m &\mathsf{if}\ \exists m.\ (a,{n},{m}) \in U\\
l &\mathsf{otherwise}
\end{cases} &\qquad\mkupd\, a\, n\,m&\mapsto \upconst{\{(a,n,m)\}}
\end{aligned}\]
\end{description}}

\vspace{-1ex}

\dlinea
\caption{the extension $\SystemT$ of G\"odel's system $\SystemTG$}\label{fig:F}
\end{figure*}
\comment{System $\SystemT$ is obtained from system $\SystemTG$ adding a new atomic type and new operations on it.}

{As proved in \cite{AB,ABF}, $\SystemT$ is strongly normalizing, has the uniqueness-of-normal-form property and the following normal form theorem also holds.

\begin{lemma}[Normal Form Property for $\SystemT$] \label{lemma-normalform} Assume $A$ is either an atomic type or a product type. Then any closed normal term $t \in \SystemT$ of type $A$ is: a numeral ${n}:\Nat$, or a boolean $\True,\False:\Bool$, or an update constant $\upconst{U}:\Update$, or a constant of type $A$, or a pair $\langle u,v \rangle: B \times C$.
\end{lemma}
}

\subsection{The System $\SystemTClass$}

\vspace{-1.5ex}

We now define a classical extension of $\SystemT$, that we call $\SystemTClass$, with a Skolem function symbol for each arithmetical formula. The elements of $\SystemTClass$ will represent (non-computable) realizers.
\begin{definition}[The System $\SystemTClass$] \label{definition-TermLanguageL1}
Define
$\SystemTClass = \SystemT+ \skolemcon$, where $\skolemcon$ is a countable set of Skolem function constants,  each one  of type $\Nat\rightarrow \Nat$. We assume to have an enumeration $\Phic_{0}, \Phic_{1}, \Phic_{2}, \ldots$ of  all the constants in $\skolemcon$ (while generic elements of $\skolemcon$ will be denoted with letters $\Phic, \Uppsi, \ldots $).
\end{definition}

Every $\Phic\in \skolemcon$ represents a \emph{Skolem function} for some arithmetical formula $\exists y^\Nat\, A(x,y)$,  taking as argument a number $x$ and returning some $y$ such that $A(x,y)$ is true if any exists, and an arbitrary value otherwise.
 In general, there is no set of computable reduction rules for the constants in $\skolemcon$, and therefore no set of computable reduction rules for $\SystemTClass$. Each (in general, non-computable) term $t \in \SystemTClass$ is associated to a set $\{t[{s}]\ | s\in \SystemT, s: \Nat^{2}\rightarrow \Nat\} \subseteq \SystemT$ of computable terms we call its ``approximations'', one for each term $s: \Nat^{2}\rightarrow \Nat $ of $\SystemT$, which is thought as a sequence $s_{0}, s_{1}, s_{2}, \ldots$ of computable approximations of the oracles $\Phic_{0}, \Phic_{1}, \Phic_{2},\ldots$ (with $s_{i}$ we denote $s(i)$).

\begin{definition}[Approximation at State]\mbox{}
\begin{enumerate}
\item A \emph{state} is a closed term of type $\Nat^{2}\redto \Nat$ of $\SystemT$. If $i$ is a numeral, with $s_{i}$ we denote $s(i)$.\\

\vspace{-1ex}

\item Assume $t\in\SystemTClass$ and $s$ is a state. The ``approximation of $t$ at a state $s$'' is the term $t[s]$ of $\SystemT$ obtained from $t$ by replacing each constant $\Phic_i$ with $s_i$.
\end{enumerate}
\end{definition}

\section{Interactive Realizability for $\HA^{\omega}+\EM+\SK$}\label{section-ALearningBasedRealizability}

\vspace{-1.5ex}

In this section we introduce a notion of realizability based on interactive learning for $\HA^{\omega} +
\EM+\SK$, Heyting Arithmetic in all finite types (see e.g. Troelstra \cite{Troelstra}) plus Excluded Middle  and Skolem axiom schemes for all
arithmetical formulas.
Then we prove our main Theorem, the Adequacy
Theorem: {\em ``if a closed  formula is provable in
$\HA^{\omega}
+ \EM +\SK$, then it is realizable''}. \IfPaperState{For proofs we
refer to
\cite{ExtendedVersion}.}{}

We first define the formal system $\HA^{\omega} + \EM +\SK$. We represent atomic predicates of $\HA^{\omega} + \EM +\SK$ with  closed terms of $\SystemTClass$ of type $\Bool$. Terms of $\HA^{\omega} + \EM+\SK$ are elements of $\SystemTClass$ and thus may include the function symbols in  $\skolemcon$. We assume having in G\"odel's $\SystemTG$ some terms $ \Rightarrow_\Bool: \Bool\rightarrow \Bool\rightarrow\Bool, \neg_\Bool: \Bool \rightarrow \Bool, \lor_{\Bool}: \Bool\rightarrow \Bool\rightarrow\Bool\ldots$, implementing boolean connectives. As usual, we shall use infix notation: for example, we write $t_{1}\Rightarrow_{\Bool} t_{2}$ in place of $\Rightarrow_{\Bool} t_{1}t_{2}$ and similarly for the other connectives.\comment{ If $t_1, \ldots, t_n, t \in \SystemTG$ have type $\Bool$ and are made from free variables all of type $\Bool$, using boolean connectives, we say that $t$ is a tautological consequence of $t_1, \ldots, t_n$ in $\SystemTG$ (a tautology if $n=0$) if all boolean assignments making $t_1, \ldots, t_n$ equal to ${\True}$ in $\SystemTG$ also make $t$ equal to ${\True}$ in
$\SystemTG$.}

\subsection{Language of $\HA^{\omega} + \EM+\SK$}\label{section-HAomega}

\vspace{-2ex}

We now define the language of the arithmetical theory $\HA^{\omega} + \EM+\SK$.

\begin{definition}[Language of $\HA^{\omega} + \EM+\SK$]\label{definition-extendedarithmetic}
The language $\LanguageClass$ of $\HA^{\omega} + \EM+\SK$ is defined as follows.
\begin{enumerate}

\item
The terms of $\LanguageClass$ are all $t \in \SystemTClass$.\\

\vspace{-1ex}

\item
The atomic formulas of $\LanguageClass$ are all $Q\in \SystemTClass$ such that $Q:\Bool$.\\

\vspace{-1ex}

\item
The formulas of $\LanguageClass$ are built from atomic formulas of $\LanguageClass$ by the connectives $\lor,\land,\rightarrow, \sub, \forall,\exists$ as usual, with quantifiers possibly ranging over variables $x^{\tau}, y^{\tau}, z^{\tau},\ldots$ of arbitrary finite type $\tau$ of $\SystemTClass$.\\

\vspace{-1ex}

\item A formula of $\LanguageClass$ is said \emph{arithmetical} if it does not contain constants in $\skolemcon$ and all its quantifiers range over the type $\Nat$, i.e. it has one of the following forms: $\forall x^{\Nat} A, \exists x^{\Nat} A, A\lor B, A\land B, A\rightarrow B, A\sub B,  P$, with $A, B$ arithmetical and $P$ atomic formula of $\SystemT$.

\end{enumerate}


\end{definition}

We denote with  $\bot$ the atomic formula ${\False}$ and with $\lnot A$ the formula $A\rightarrow \bot$. $A\sub B$ is the dual of implication as in bi-intuitionistic logic and means ``$A$ and the opposite of $B$''.  If $F$ is a formula of $\LanguageClass$ in the free variables $x_{1}^{\tau_{1}},\ldots, x_{n}^{\tau_{n}}$ and $t_{1}:\tau_{1}, \ldots, t_{n}:\tau_{n}$ are terms of $\LanguageClass$, with $F(t_{1}, \ldots, t_{n})$ we shall denote the formula $F[t_{1}/x_{1}, \ldots, t_{n}/x_{n}]$. Sequences of variable $x_{1}^{\Nat}, \ldots, x_{k}^{\Nat}$ will be written as $\vec{x}$. We denote with $\code{\vec{x}}$ a term of $\SystemT$ in the free numeric variables $\vec{x}$ representing a injection of $\Nat^{k}$ into $\Nat$. Moreover, for every sequence of numerals $\vec{n}=n_{1},\ldots, n_{k}$, we define $\code{\vec{n}}:=\code{\vec{x}}[\vec{n}/\vec{x}]$ and assume that the function $\vec{n}\mapsto \code{\vec{n}}$ is a bijection.

The \emph{Excluded Middle axiom scheme} $\EM$ is defined as the set of all formulas of the form:
\[\forall \vec{x}^{\,\Nat}.\ A(\vec{x})\vee \neg A(\vec{x})\]
where $A$ is an arithmetical formula.

The \emph{Skolem axiom scheme} $\SK$ contains for each arithmetical formula $A(\vec{x},y)$
 an axiom: \[\forall \vec{x}^{\,\Nat}.\ \exists y^\Nat A(\vec{x},y)\rightarrow A(\vec{x}, \Phic\code{\vec{x}})\]
with $\Phic\in\skolemcon$. We assume that for every $\Phic\in\skolemcon$ there is in $\SK$ one and only one formula in which $\Phic$ occurs.  Such unique formula $A$ is said to be the \emph{formula associated to} $\Phic$ and $\Phic$ will be sometimes written as $\Phic_{A}$. If $s$ is a state and $\Phic_{i}=\Phic_{A}$, with $s_{A}$ we denote $s_{i}$ and with $\mkupd\,A\, u\, t$ we denote $\mkupd\, i\,u\,t$. We claim that the result of this paper would even hold if the formula $A$ was not required to be arithmetical, i.e. it was allowed to contain other Skolem functions previously defined by other Skolem axioms, possibility which in Avigad's case \cite{AvigadSkolem} complicates the elimination technique considerably.

\comment{ With $\lev{\Phic}$ we denote the number measuring the logical complexity of the formula $A$ associated to $\Phic$, i.e. the number of logical symbols occurring in $A$. If $s, s'$ are two states  and $n$ a numeral, we write $s\equiv s'\, \lev{n}$ if for every $A$ of complexity less than $n$, one has $s_{A}(m)=s'_{A}(m)$ for all numerals $m$. }

For each formula $F$ of $\LanguageClass$, its involutive negation $F^{\bot}$ is defined  by induction on $F$. First, we say that an atomic formula $P$ is positive if it is of the form $\lnot_{\Bool}\ldots \lnot_{\Bool} Q$, $Q$ is not of the form $\lnot_{\Bool} R$, and the number of $\lnot_{\Bool}$ in front of $Q$ is even.  Then we define:
\[\begin{aligned}
(\lnot_{\Bool}P)^{\bot}&= P\text{ (if $P$ positive)} &\qquad P^{\bot}&=\lnot_{\Bool}P \text{ (if $P$  positive)}\\
(A\wedge B)^{\bot}&=A^{\bot}\lor B^{\bot}&\qquad (A\vee B)^{\bot}&= A^{\bot} \land B^{\bot}\\
(A\rightarrow B)^{\bot}&=A\sub B&\qquad (A\sub B)^{\bot} &= A\rightarrow B\\
 (\forall x^{\tau} A)^{\bot}&=\exists x^{\tau}A^{\bot}&\qquad  (\exists x^{\tau} A)^{\bot}&= \forall x^{\tau} A^{\bot}\\
\end{aligned}
\]
As usual, one has $(F^{\bot})^{\bot}=F$.
\comment{We defined $\Rightarrow_\Bool:\Bool,\Bool\rightarrow \Bool$ as a term implementing implication, therefore, to be accurate, formulas of the form $\empredicate{a}(t,u) \Rightarrow_\Bool \empredicate{a}(t,\Phic_{a} t)$ are not an implication between two atomic formulas, but  they are equal to the single atomic formula $Q$, where \[ Q:=\   \Rightarrow_\Bool (\empredicate{a}tu)(\empredicate{a}t(\Phic_{a} t))\]}


We now fix  a special set of formulas  $\FSet$.
\begin{definition}[Set $\FSet$]\label{}
We fix  an arbitrary  finite set $\FSet$ of arithmetical formulas $A(\vec{x},y)$ of $\LanguageClass$.
\end{definition}

In the following, $\FSet$ will serve as a parameter in order to relativize the definitions of the realizability relation and of the ordering of states provided in \cite{AschieriPA}.
 The idea is that any given proof in the system $\HAw+\EM+\SK$ uses only a finite number of  instances of $\EM$ and $\SK$. Thus, it is enough to specialize the atomic case of the definition of realizability in such a way it refers only to the formulas in $\FSet$.
 The restriction  is necessary in order to avoid to speak about the truth of an infinite number of formulas, as done in \cite{AschieriPA}. When we shall have to interpret a particular proof $P$,
we will choose $\FSet$ as containing all the sub-formulas of the classical axioms appearing in $P$.

\subsection{Truth Value of a Formula in a State}

\vspace{-1.5ex}

The axioms of the system $\HA^{\omega}+\EM+\SK$ give a great computational power to the system $\SystemTClass$: thanks to the use of Skolem functions as oracles, one can ``compute'' by a term $\truth{F}$ of $\SystemTClass$ the truth value  of any arithmetical formula $F$. When one effectively evaluates $\truth{F}$ in a particular state $s$, we say that one computes \emph{the truth value of a formula $F$ in a state $s$}.

 \begin{definition}[Truth Value of a Formula $F$ in a State $s$]\label{definition-truthstate}
 For every arithmetical formula $F(\vec{x})$ of $\LanguageClass$ we define, by induction on $F$, a term $\truth{F}:\Bool$ of system $\SystemTClass$, with the same free variables of $F$:

 \vspace{-3ex}

\[\begin{aligned}
\truth{P}&=P,\ P\text{ atomic}\\
\truth{A\lor B}&=\truth{A}\lor_{\Bool} \truth{B}\qquad & \truth{\forall y^{\Nat} A}&=\truth{A}[\Phic_{ A^{\bot}} \code{\vec{x}}/y]\qquad &  \truth{A\sub B}&=\truth{A}\land_{\Bool} \truth{B^{\bot}}\\
\truth{A\land B}&=\truth{A}\land_{\Bool} \truth{B} \qquad& \truth{\exists y^{\Nat} A}&=\truth{A}[\Phic_{A}\code{\vec{x}}/y]\qquad &\truth{A\rightarrow B}&=\truth{A}\Rightarrow_{\Bool}\truth{B}
\end{aligned}
\]
We define $\truthst{F}{s}:=\truth{F}[s]$ and call it \emph{the truth value of $F$ in the state $s$}.
 \end{definition}
Intuitively, if $F(\vec{n})$ is a closed formula, our intended interpretation is:

\begin{enumerate}\item $\truth{F}(\vec{n})$ is a term of $\SystemTClass$ \emph{denoting}, in any standard model of $\HA^{\omega}+\EM+\SK$, the truth value of $F(\vec{n})$. \item  $\truthst{F}{s}(\vec{n})$ is a term of $\SystemT$  \emph{computing what would be} the truth value of $F(\vec{n})$ in some standard model of $\HA^{\omega}+\EM$ \emph{under the (possible false) assumption} that the interpretation mapping $\Phic_{i}$ to $s_{i}$ satisfies the axioms of $\SK$.\end{enumerate}

We remark that thus $\truthst{F}{s}(\vec{n})$ is only a \emph{conditional} truth value: if $\truthst{F}{s}(\vec{n})$ is not the correct truth value of $F(\vec{n})$ --  it may well happen -- then  the interpretation mapping $\Phic_{i}$ in $s_{i}$ does not satisfy the axioms of $\SK$. This subtle point is what makes possible learning in Interactive realizability: whenever a contradiction follows, realizers are able to effectively find counterexamples to the assertion that the interpretation mapping  $\Phic_{i}$ in  $s_{i}$ satisfies the axioms of $\SK$. We also observe that this way of computing the truth of a formula comes from the epsilon substitution method (see Avigad \cite{Avigad}, Mints et al. \cite{Mints}).

\comment{
The notion of truth in a state behaves as expected with respect to involutive negation.

\begin{proposition}[Truth in a State and Truth]\label{proposition-negation}
\comment{\begin{enumerate}
\item For every arithmetical formula $F$
\[\HA^{\omega}+\EM+\SK\vdash F\leftrightarrow \truth{F}\]
\item}
For every arithmetical formula $F(\vec{x})$, state $s$ and sequence of numerals $\vec{n}$,
$\truthst{F}{s}(\vec{n})=\False \iff \truthstbot{F}{s}(\vec{n})=\True$.
\end{proposition}

The truth of a formula $F$ in a state $s$ is determined by the approximations that $s$ gives of the Skolem functions of strictly lower level than the logical complexity of $A$.

\begin{proposition}\label{proposition-wf}
Let $F(\vec{x})$ be any arithmetical formula of logical complexity $m$, $\vec{n}$ be a sequence of numerals  and $s, s'$ be states such that $s\equiv s'\, \lev{m}$. Then $\truthst{F}{s}(\vec{n})=\truthst{F}{s'}(\vec{n})$.
\end{proposition}
}
 Every state $s$ is considered as an \emph{approximation} of the Skolem functions denoted by the constants of $\skolemcon$: for each formula $A$, $s_{A}$ may be a correct approximation of  $\Phic_{A}$ on some arguments, but wrong on other ones. More precisely,  we are going to consider the set $\dm{s}$ of the pairs $(i,\code{\vec{n}})$ such that $\Phic_{i}=\Phic_{A}$ and $A\in \FSet\Rightarrow\exists y^{\Nat} A(\vec{n},y)\redto{A}(\vec{n}, s_{i}\code{\vec{n}})$ is true as the real ``domain'' of $s$, representing the set of arguments at which $s_{i}$ is surely a correct approximation of $\Phic_{i}$, in the sense that $s_{i}$ returns an appropriate witness if any exists.
 \comment{On the other hand, we also need to consider an approximated domain $\dom{}{s}$ which comprises all the pairs $(i,\code{\vec{n}})$ such that $\Phic_{i}=\Phic_{A}$ and $A\in \FSet \Rightarrow\truthst{A}{s}(\vec{n}, s_{i}\code{\vec{n}})=\True$ (observe that the truth of $A$ is now approximated in $s$).
  We point out that in both cases if $\Phic_{i}=\Phic_{A}$ and $A\notin \FSet$, then trivially $(i,\code{\vec{n}})\in \dm{s}, \dom{}{s}$.
 }
 We point out that if $\Phic_{i}=\Phic_{A}$ and $A\notin \FSet$, then trivially $(i,\code{\vec{n}})\in \dm{s}$.
 The choice is made just for technical convenience, since one is not interested in the behaviour of $s$ outside $\FSet$.  We also define an ordering between states: we say that $s'\geq s$ if, intuitively, $s'$ is at least as good an approximation as $s$. Thus, we ask that if $s$ is a correct approximation at argument $(i,\code{\vec{n}})$ also $s'$ is and in particular $s'_{i}\code{\vec{n}}=s_{i}\code{\vec{n}}$.

  \begin{definition}[Domains, Ordering between States\comment{Sound Updates}]\label{definition-soundupdates}\mbox{}
\begin{enumerate}

\item We define
$$\dm{s}=\{(i,\code{\vec{n}})\;|\; \Phic_i=\Phic_{A} \text{ and $(A\in\FSet\Rightarrow \exists y^{\Nat} A(\vec{n},y)\redto A(\vec{n}, s_{i}\code{\vec{n}})$}\}$$ where $i$ and $\vec{n}$ range over numerals and sequences of numerals.\\
\item Let $s$ and $s'$ be two states. We define $s'\geq s$ if and only if for all $(i,\code{\vec{n}})$, $(i,\code{\vec{n}})\in\dm{s}$ implies  $s_i\code{\vec{n}}= s'_i\code{\vec{n}}$.

\comment{
\item Given an update ${U}$ and a state $s$, we define $\dom{s}{U}$ as  the set of pairs of numerals $(i,\code{\vec{n}})$ such that $A(\vec{x},y)$ is the formula associated to $\Phic_{i}$, $A\in \FSet$,  $(i,\code{\vec{n}}, m)\in U$ and $\truthst{A}{s}(\vec{n},m)=\True$. $U$ is said to be \emph{sound in the state $s$} if $(i,\code{\vec{n}}, m)\in U$  implies  $(i,\code{\vec{n}})\in\dom{s}{U}$.\\
\item Similarly, if $s$ is a state, we denote with $\dom{}{s}$  the set of pairs of numerals $(i,\code{\vec{n}})$ such that  $A(\vec{x},y)$ is the formula associated to $\Phic_{i}$, $s_{i}\code{\vec{n}}=m$ and $A\in \Gamma \Rightarrow \truthst{A}{s}(\vec{n},m)=\True$.
}
\end{enumerate}

\end{definition}

We remark that by definition, $s'\geq s$ implies $\dm{s'}\supseteq\dm{s}$ and that thanks to the restriction to $\FSet$ the relation $s'\geq s$ is arithmetical, because the condition $(i,\code{\vec{n}})\in\dm{s}$ is non-trivial only for finitely many $i$.  From now onwards, for every pair of terms $t_1, t_2$ of system $\SystemT$, we shall write $t_1=t_2$ if they are the same term modulo the equality rules corresponding to the reduction rules of system $\SystemT$ (equivalently, if they have the same normal form).

\subsection{Interactive Realizability}

\vspace{-1.5ex}

For every formula $A$ of $\LanguageClass$, we  now define what type $|A|$ a realizer of $A$ must have.
\begin{definition}[Types for realizers]
\label{definition-TypesForRealizers} For each
 formula $A$ of $\LanguageClass$ we define a type $|A|$ of $\SystemTClass$ by
induction on $A$:
\[\begin{aligned} |P|&=\Update, \text{ if $P$ is atomic}\\
|A\wedge B|&=|A|\times |B| &\qquad |\exists x^{\tau} A|&= \tau\times |A|\qquad & |A\sub B|&= |A| \times |B^{\bot}|\\
|A\vee B|&= \Bool\times (|A| \times|B|) &\qquad |\forall x^{\tau} A|&=\tau\rightarrow |A|\qquad &|A\rightarrow B|&=|A|\rightarrow |B|
\end{aligned}
\]
\end{definition}

Let now $\proj_0:=\pi_0 : \sigma_0 \times (\sigma_1 \times \sigma_2) \rightarrow \sigma_0 $, $\proj_1:=\pi_0\pi_1 : \sigma_0 \times (\sigma_1 \times \sigma_2) \rightarrow \sigma_1$ and $\proj_2:=\pi_1\pi_1 : \sigma_0 \times (\sigma_1 \times \sigma_2) \rightarrow \sigma_2$ be the three canonical projections from $\sigma_0 \times (\sigma_1 \times \sigma_2)$. We define the realizability relation $t\Vvdash F$, where $t \in \SystemTClass$, $F \in \LanguageClass$ and $t:|F|$.

\begin{definition}[Interactive Realizability]
\label{definition-IndexedRealizabilityAndRealizability}
Assume $s$ is a state, $t$ is a closed term of $\SystemTClass$, $F \in \LanguageClass$ is a closed formula, and $t:|F|$. We define first the relation $t\Vvdash_s F$ by induction and by cases according to the form of $F$:

\begin{enumerate}
\item
$t\Vvdash_s Q$ for some atomic $Q$ if and only if $\upconst{U} = t[s]$ implies:
\begin{itemize}
\item
for every  $(i, \vec{n}, m)\in U$,  $\Phic_i=\Phic_A$ for some $A\in\FSet$, and $\truthst{A}{s}(\vec{n}, s_i \code{\vec{n}})=\False$ and $\truthst{A}{s}(\vec{n},m)=\True$.

\item
$\upconst{U}  = \varnothing$  implies $Q[{s}]={\True}$ \\

\end{itemize}

\vspace{-1ex}

\item
$t\Vvdash_s{A\wedge B}$ if and only if $\pi_0t \Vvdash_s{A}$ and $\pi_1t\Vvdash_s{B}$\\

\vspace{-1ex}

\item
$t\Vvdash_s {A\vee B}$  if and only if either $\proj_0t[{s}]={\True}$ and $\proj_1t\Vvdash_s A$, or $\proj_0t[{s}]={\False}$ and $\proj_2t\Vvdash_s B$ \\

\vspace{-1ex}

\item
$t\Vvdash_s {A\rightarrow B}$ if and only if for all $u$, if $u\Vvdash_s{A}$,
then $tu\Vvdash_s{B}$\\

\vspace{-1ex}

\item $t\Vvdash_{s} A\sub B$ if and only if $\pi_{0}t \Vvdash_{s} A$ and $\pi_{1}t\Vdash_{s} B^{\bot}$\\

\vspace{-1ex}

\item
$t\Vvdash_s {\forall x^{\tau} A}$ if and only if for all closed terms $u:\tau$ of $\SystemT$,
$tu\Vvdash_s A[{u}/x]$\\
\item

\vspace{-1ex}

$t\Vvdash_s \exists x^{\tau} A$ if and only for some closed term $u:\tau$ of $\SystemT$, $\pi_0t[{s}]= u$  and $\pi_1t \Vvdash_s A[{u}/x]$\\

\end{enumerate}

\vspace{-1ex}
We define $t\Vvdash F$ if and only if for all states $s$ of $\SystemT$, $t\Vvdash_s F $.

\end{definition}

The ideas behind the definition of $\Vvdash_s$ in the case of $\HA^{\omega}+\EM+\SK$ are those we already explained in
\cite{AschieriPA}. A realizer is a term $t$ of $\SystemTClass$, possibly containing some non-computable Skolem function of $\skolemcon$; if such a function was computable, $t$ would be an intuitionistic realizer. Since in general $t$ is not computable, we calculate its approximation $t[s]$ at state $s$. $t$ is an intelligent, self-correcting program, representing a proof/construction depending on the state $s$.  The realizer \emph{interacts} with the environment, which may provide a counter-proof, a counterexample invalidating the current construction of the realizer. But the realizer is always able to turn such a negative outcome into a positive information, which consists in some new piece of knowledge learned about some Skolem function $\Phic_{i}$.

The next proposition tells that realizability at state $s$ respects the notion of equality of $\SystemTClass$  terms, when the latter is  relativized to state $s$. That is, if two terms are equal at the state $s$, then they realize the same formulas in the state $s$.

\begin{proposition}[Saturation]\label{proposition-saturation}
If $t_{1}[s]=t_{2}[s]$ and $u_{1}[s]=u_{2}[s]$, then $ t_{1}\Vvdash_s B[u_{1}/x]$ if and only if $t_{2}\Vvdash_s B[u_{2}/x]$.

\end{proposition}
{\textit{Proof}.
By straightforward induction on $A$.
}

\comment{
\begin{proposition}[Truth and Realizability in a State]\label{proposition-truthstate} Let $F(\vec{x})$ be an arithmetical formula. There exist two terms $\real{F}(\vec{x})$ and $\freal{F}(\vec{x})$ of $\SystemTClass$ such that for all numerals $\vec{n}$ and state $s$
\begin{equation}\label{eq1}\truthst{F}{s}(\vec{n})=\True \implies \real{F}(\vec{n})\Vvdash_{s} F(\vec{n})\end{equation}
\begin{equation}\label{eq2}\truthst{F}{s}(\vec{n})=\False \implies \freal{F}(\vec{n})\Vvdash_{s} \lnot F(\vec{n})\end{equation}
\end{proposition}
}

In the following, we use a standard natural deduction system for $\HA^{\omega}+\EM+\SK$, together with a term assignment  in the spirit of Curry-Howard correspondence for classical logic. We denote with  $\HA^{\omega} + \EM +\SK\vdash t: A$ the derivability relation in that system, where $t$ is a term of $\SystemTClass$ and $A$ is a formula of $\LanguageClass$. All details can be found in ~\cite{ABF},~\cite{AschieriPA}.

The main theorem about Interactive realizability is the Adequacy Theorem: if a closed formula is provable in $\HA^{\omega} + \EM +\SK$, then it is realizable (see \cite{AschieriPA} for a proof).

\begin{theorem}[Adequacy Theorem]\label{Realizability Theorem}
If $A$ is a closed formula such that $\HA^{\omega} + \EM +\SK\vdash t: A$ and all the subformulas of the instances of $\EM$ and $\SK$ used in the derivation belong to $\FSet$, then $t\Vvdash A$.
\end{theorem}

\section{Conservativity of $\HA^{\omega}+\EM+\SK$ over $\HA^{\omega}+\EM$ ($\HA+\EM$)}\label{section-Conservativity}

\vspace{-1.5ex}

The aim of  this section is to use Interactive realizability in order to prove that for every arithmetical formula $A$,  if $\HAw+\EM+\SK\vdash A$ then  $\HAw+\EM\vdash A$ ($\HA+\EM\vdash A$).
Since we know by the Adequacy Theorem~\ref{Realizability Theorem} that  $\HAw+\EM+\SK\vdash A$ implies $\HA^{\omega} + \EM +\SK\vdash t: A$ and $\HAw$ proves $t\Vvdash A$, our goal is to show in $\HA^{\omega} + \EM$ that $t\Vvdash A$ implies $A$.

The intuitive reason this latter result is true is the following: one can always find an approximation $s$ of the Skolem functions of $t$ which is good enough to contain all the information needed by $t$ to compute the \emph{true} witnesses for $A$ against any particular purported counterexample.
The idea is that one has only to collect finitely many values of each Skolem function called during the execution of the program represented by $t$. To this end, it suffices to invoke the excluded middle a number of times which, intuitively, can be expressed  in a proof formalizable in $\HAw+\EM$. This is possible because   $\HAw+\EM$ is strong enough to prove the normalization of each term $t$ of $\SystemTClass$ with respect to any interpretation of its Skolem functions.
Finally, if there existed a counterexample to $A$, it would be possible to falsify the construction of the realizer $t$ in the state $s$. Since $t$ is a self-correcting program, it would be able to correct one of the values of $s$ it has used in the computation of some witness for $A$. But $s$ is constructed as to be correct on all the values used by $t$, which entails a contradiction.

For example, let $A=\exists x^{\Nat}\forall y^{\Nat}\exists z^{\Nat} P(x,y,z)$. Then one can find a state $s$ which contains all the values of the Skolem functions needed to compute $n=\pi_{0}t[s]$. Suppose a counterexample $m$ to the formula $\forall y^{\Nat}\exists z^{\Nat} P(n,y,z)$ existed. Then one can find a state $s'\geq s$ which contains all the values of the Skolem functions needed to compute $l=\pi_{0}\left((\pi_{1}t)m\right)[s']$. Now, we would have that $P(n,m,l)$ is false; thus, $\pi_{1}\left((\pi_{1}t)m\right)[s']$ would be equal to some update $\upconst{U}$ containing some corrections to $s'$. We shall show that this will not be the case, and the intuitive reason is that $s'$ can be chosen as to be correct everywhere it is needed.

We now elaborate our argument. We start with a definition axiomatizing the informal concept that a state $s$ contains all the information needed to compute the normal form of a term $t$ of ground type. Namely, if for every $s'$ extending $s$ the evaluation of $t$ in the state $s'$ gives the same result obtained evaluating $t$ in $s$, then we may assume all the relevant information is already in $s$.

\begin{definition}[Definition of a term in a state $s$]\label{definition-defstate} For every state $s$ and  term $t$ of $\SystemTClass$ of atomic type, we define  $t\tdef{s}$ (and we say ``$t$ is defined in $s$'') as the statement: for all states $s'\geq s$, $t[s']=t[s]$.
\end{definition}

\textbf{Remark}. There is another, perhaps more intuitive way to express the concept of ``being defined in the state $s$''.
For every state $s$ we may define a binary reduction relation $\rs{s}\,\subseteq \SystemTClass\times\SystemTClass$ as follows:
$t\rs{s}u$ if either $t\mapsto u$ in $\SystemTClass$ or $u$ is obtained from $t$ by replacing one of its subterms $\Phic_{i}(n)$  with a numeral  $m=s_{i}(n)$ such that $(i,n)\in\dm{s}$.
Then one could say that $t$ is defined in $s$ if $t\rs{s}a$ where $a$ is  either a numeral, a boolean or an update. Though this approach works well, it is unsuitable to be directly formalized in $\HAw$, because in that system one cannot express this syntactical reasoning on terms.

We now define for every type $\tau$ a set of ``computable'' terms of type $\tau$ by means of the usual Tait-style computability predicates \cite{Tait}. In our case, following the approach of the previous discussion, we consider a term $t$ of ground type to be computable if for every state $s$, one can find a state $s'\geq s$ such  that $t$ is defined in $s'$. The notion is lifted to higher types as usual.

\begin{definition}[Computable terms]\label{definition-Computability}\mbox{}

For every type $\tau$ of $\SystemTClass$, we define a set of closed terms of $\SystemTClass$ of type $\tau$ as follows:

\begin{itemize}
\item $\comp{\Nat}$=$\{t:\Nat\;|$  \mbox{for all states  $s$ there  is a state  $s'\geq s$ such that $t\tdef{s'}$} $\}$ \\

\vspace{-1ex}

\item $\comp{\Bool}$=$\{t:\Bool\;|$  \mbox{for all states  $s$ there  is a state  $s'\geq s$ such that  $t\tdef{s'}$} $\}$ \\

\vspace{-1ex}

\item $\comp{\Update}$=$\{t:\Update\;|$  \mbox{for all states  $s$ there  is a state  $s'\geq s$ such that  $t\tdef{s'}$} $\}$ \\

\vspace{-1ex}

\item $\comp{\tau\redto \sigma}$=$\{t\;|\; \forall u\in\comp{\tau}\hspace{-3pt}. \;tu\in\comp{\sigma}\} $ \\

\vspace{-1ex}

\item $\comp{\tau\times \sigma}$=$\{t\;|\; \pi_{0}t\in\comp{\tau} \mbox{and } \pi_{1}t\in\comp{\sigma}\} $

\end{itemize}
\end{definition}

\comment{
\begin{proposition}[Monotonicity]\label{prop:monotonicity}
For all states $s$, if $t\tdef{s}$ and $s'\geq s$, then $t\tdef{s'}$.
\end{proposition}
\begin{proof}
Let $s$ be a state and $s'\geq s$. Suppose that $s''\geq s'$. Then $s''\geq s$ and since $t\tdef{s}$, we have $t[s'']=t[s]$, $t[s']=t[s]$ and therefore $t[s'']=t[s']$. We conclude $t\tdef{s'}$.
\end{proof}

}

In order to show that every term $t$ in $\SystemTClass$ is computable, as usual we need to prove that  the set of computable terms is saturated with respect to some suitable relation. In our case, two terms are related if they are equal in all states greater than some state.

\begin{lemma}\label{lemma:saturation}
 For every term $t:\rho$ of $\SystemTClass$, if for every state $s$ there exists a state $s'\geq s$ and $u\in\comp{\rho}$ such that for all state $s''\geq s'$, $t[s'']=u[s'']$, then $t\in\comp{\rho}$.
\end{lemma}

\textit{Proof}.
By induction on the type $\rho$.

\begin{itemize}
\item $\rho =\Nat$.  Let $s$ be a state. We have to show that there exists a state $r\geq s$ such that $t\tdef{r}$.
By assumption on $t$ there exists a state $s'\geq s$ and  $u\in\comp{\Nat}$ such that for all $s''\geq s'$, $t[s'']=u[s'']$. Since $u\in\comp{\Nat}$,  there exists $s''\geq s'$ such that $u\tdef{s''}$. Let $r=s''$; we prove  $t\tdef{r}$. Let  $r'\geq r$. We have that $u[r']=u[r]$,  by $u\tdef{r}$, and $t[r']=u[r']$, since $r'\geq s'$. Hence, $t[r']=u[r]=t[r]$. We conclude $t\tdef{r}$ and finally $t\in\comp{\Nat}$.\\

\vspace{-1ex}

\item $\rho=\Bool,\Update$: as for the case $\rho=\Nat$.\\

\vspace{-1ex}

\item $\rho=\tau\redto \sigma$. Let $v\in\comp{\tau}$. We have to show that $tv\in\comp{\sigma}$. Let $s$ be any state. By assumption on $t$ there exist a state $s'\geq s$ and  $u\in\comp{\tau\redto\sigma}$ such that for all $s''\geq s'$, $t[s'']=u[s'']$. Therefore for all $s''\geq s'$, $tv[{s''}]=uv[s'']$ and $uv\in\comp{\sigma}$. Hence, by induction hypothesis, $tv\in\comp{\sigma}$.\\

\vspace{-1ex}

\item $\rho=\tau_0\times \tau_1$. Let $i\in\{0,1\}$, we have to show that $\pi_i t\in\comp{\tau_i}$. Let $s$ be any state. By assumption on $t$ there exist $s'\geq s$  and $u\in\comp{\tau_0\times \tau_1}$ such that  for all $s''\geq s'$, $t[{s''}]=u[s'']$. Therefore for all $s''\geq s'$, $\pi_{i}t[{s''}]=\pi_{i}u[s'']$ and $\pi_{i}u\in\comp{\tau_i}$. Hence, by induction hypothesis $\pi_{i}t\in\comp{\tau_i}$.
\end{itemize}

We are now ready to prove, by using the excluded middle alone, that every term $t$ of $\SystemTClass$ is computable.

{
\begin{theorem}[Computability Theorem]\label{theorem-computability}\mbox{}

\noindent
Let $v: \tau$ be a term of $\SystemTClass$ and suppose that all the free variables of $v$ are among $x_1^{\sigma_1},\ldots, x_{n}^{\sigma_n}$.
If $t_{1}\in\comp{\sigma_1}, \ldots, t_{n}\in\comp{\sigma_n}$, then $v[t_{1}/x^{\sigma_1}_{1},\ldots, t_{n}/x^{\sigma_n}_{n}]\in\comp{\tau}$.
\end{theorem}

\textit{Proof}. We proceed by induction on $v$. We first remark that if $u=t$ and  $t\in\comp{\tau}$, then $u\in \comp{\tau}$ by trivial application of Lemma \ref{lemma:saturation}.
\begin{notation}
 For any term $w$ in $\SystemTClass$, we denote $w[t_1/x_1^{\sigma_1},\ldots, t_n/x^{\sigma_n}]$ with $\overline{w}$.
\end{notation}

\vspace{-2ex}

\begin{enumerate}
\item $v$ is a variable $x_i^{\sigma_{i}}:\sigma_{i}$ and $\tau=\sigma_i$. Then, $\sbs{v}=t_1\in\comp{\sigma_i}=\comp{\tau}$.\\

\vspace{-1ex}

\item $v$ is $0$,  $\True$, $\False$, $\upconst{U}$: trivial.\\

\vspace{-1ex}

\item  $v$ is $uw$, then by means of typing rules, $u:\sigma\redto\tau$, $w:\sigma$. Since by induction hypothesis $\overline{u}\in\comp{\sigma\redto\tau}$ and  $\overline{w}\in\comp{\sigma}$, we obtain $\overline{v}=\overline{u}\overline{w}\in\comp{\tau}$.\\

\vspace{-1ex}

\item $v$ is $\lambda x^{\tau_{1}}. u:\tau_{1}\redto\tau_{2}$. Then, by means of typing rules, $u:\tau_{2}$. Suppose now, for a term $t:\tau_1$ in $\SystemTClass$, that $t\in\comp{\tau_1}$. We have to prove that $\sbs{v}t\in\comp{\tau_2}$. We have:

\comment{
\[\begin{aligned}
\sbs{v}t &=
(\lambda x^{\tau_1}. u)[t_1/x_1^{\sigma_1}\cdots t_n/x_n^{\sigma_n} ] t
\\
&= (\lambda x^{\tau_1} u)t[t_1/x_1^{\sigma_1}\cdots t_n/x_n^{\sigma_n}]
\\
&=u[t/x^{\tau_1}][t_1/x_1^{\sigma_1}\cdots t_n/x_n^{\sigma_n}]
\\
&=u[t/x^{\tau_1}\, t_1/x_1^{\sigma_1}\cdots t_n/x_n^{\sigma_n}]
\\
\end{aligned}
\]
}

\[\begin{aligned}
\sbs{v}t &=
(\lambda x^{\tau_1}. u)[t_1/x_1^{\sigma_1}\cdots t_n/x_n^{\sigma_n} ] t \\
&= (\lambda x^{\tau_1} u)t[t_1/x_1^{\sigma_1}\cdots t_n/x_n^{\sigma_n}]
\mbox{} =u[t/x^{\tau_1}][t_1/x_1^{\sigma_1}\cdots t_n/x_n^{\sigma_n}]
=u[t/x^{\tau_1}\, t_1/x_1^{\sigma_1}\cdots t_n/x_n^{\sigma_n}]\end{aligned}
\]

\bigskip

By induction hypothesis, this latter term belongs to $\comp{\tau_2}$.
We conclude $\sbs{v}t\in\comp{\tau_2}$.\\

\vspace{-1ex}

\item $v$ is $\langle u,w\rangle:\tau_0\times\tau_1$. By means of typing rules, $u:\tau_0$, $w:\tau_1$ and by induction hypothesis $\pi_0 \substitution{v}=\sbs{u}\in\comp{\tau_0}$ and $\pi_1 \substitution{v} =\sbs{w}\in\comp{\tau_1}$. The thesis  $\sbs{v}\in\comp{\tau_0\times\tau_1}$ follows by definition.\\

\vspace{-1ex}

\item $v$ is $\pi_i(u):\tau_i$, $i=0,1$, where $u: \tau_0\times \tau_1$.  ${\pi_i \sbs{u}}\in\comp{\tau_i}$ because $\sbs{u}\in\comp{\tau_0\times \tau_1}$ by induction hypothesis. \\

\vspace{-1ex}

\item $v$ is $\ifn_{\tau}:\Bool\rightarrow \tau\rightarrow  \tau\rightarrow \tau$. Suppose that $u\in\comp{\Bool}$, $u_1\in\comp{\tau}$, $u_2\in\comp{\tau}$. Then, for all states $s$ there exists $s'\geq s$ such that $u\tdef{s'} $. We have to prove that $\ifn_{\tau} u u_{1} u_{2}\in{\comp{\tau}}$.
 Let $s$ be a state and let $s'\geq s$ be such that $u\tdef{s'}$. If $u[s']=\True$, then for all $s''\geq s'$, $\ifn_{\tau} u u_{1} u_{2}[s'']=u_1[s'']$ and ${u_1}\in\comp{\tau}$. If $u[s']=\False$, then for all $s''\geq s'$, $\ifn_{\tau} u u_{1} u_{2}[s'']=u_2[s'']$ and ${u_2}\in\comp{\tau}$.  By Lemma~\ref{lemma:saturation}, we conclude $\ifn_{\tau} u u_{1} u_{2}\in{\comp{\tau}}$.  \\

\vspace{-1ex}

\item $v$ is $\rec_{\tau}: \tau \rightarrow (\Nat \rightarrow (\tau \rightarrow \tau))\rightarrow \Nat\rightarrow \tau$. Suppose that $u\in\comp{\tau}$, $w\in\comp{\Nat \rightarrow (\tau \rightarrow \tau)}$, $z\in\comp{\Nat}$. We have to prove that $\rec_{\tau}\;uwz\in\comp{\tau}$.
By a plain induction, it is possible to prove, for each numeral $n$, $\rec_{\tau}\;uwn\in\comp{\tau}$.
Let $s$ be a state and let $s'\geq s$ be such that $z\tdef{s'}$.  Let $z[s']=n$ with $n$ numeral. Then for all $s''\geq s'$, $$\rec_{\tau}uvz[s'']=\rec_{\tau}uvn[s'']\in\comp{\tau}$$ By Lemma~\ref{lemma:saturation}, we conclude $\rec_{\tau}uwz\in\comp{\tau}$.\\

\vspace{-1ex}

\item $v$ is $\minu:\Update\rightarrow \Nat$. Suppose, for a term $u$ in $\SystemTClass$, that $u\in\comp{U}$. Let $s$ be a state. Since $u\in\comp{U}$, there exists $s'\geq s$ such that $u\tdef{s'}$.  We have to prove that $\minu \;u\in{\comp{\Nat}}$. There exists an update $\upconst{U}$ such that for all $s''\geq s'$, $u[s'']=\upconst{U}$. Then for all $s''\geq s'$, $\min u[s'']=\min \upconst{U}=n$ for some numeral $n$. By definition of $\comp{\Nat}$, $\minu \;u\in{\comp{\Nat}}$.\\

\vspace{-1ex}

\item $v$ is $\Cup: \Update\rightarrow \Update\rightarrow \Update$. Suppose that $u_1\in\comp{\Update}$ and $u_2\in\comp{\Update}$. We have to prove that $\Cup\; u_1 u_2\in\comp{\Update}$. Let $s$ be a state. Since $u_1\in\comp{\Update}$ there exists $s'\geq s$  such that $u_1\tdef{s'}$.
Since $u_2\in\comp{\Update}$,  there exists $s''\geq s'$ such that $u_2\tdef{s''}$. Therefore, there exist two constants $\upconst{U}_{1}$ and $\upconst{U}_{2}$ such that for all $s'''\geq s''$, $u_1[s''']=\upconst{U}_{1}$ and  $u_2[s''']=\upconst{U}_2$. Finally, for all $s'''\geq s''$, $$\Cup\; u_1 u_2[s''']=\Cup\; \upconst{U}_1 \upconst{U}_2=\upconst{U}_3$$ and by definition of $\comp{\Update}$, $\Cup\; u_1 u_2\in\comp{\Update}$.\\

\vspace{-1ex}

\item $v$ is $\suc$, $\mkupd$ or  $\get$. Analogous to the previous case.\\

\vspace{-1ex}

\item $v$ is a constant  $\Phic_{i}:\Nat\redto\Nat$ in $\skolemcon$. Suppose now, for a term $u:\Nat$, that  $u\in\comp{\Nat}$. We have to prove that $\Phic_i u\in\comp{\Nat}$. Let $s$ be a state. We must show that there exists a $s'\geq s$ such that $\Phic_{i}u \tdef{s'}$. Since $u\in\comp{\Nat}$,  there exists a state $s' \geq s$ such that $u\tdef{s'}$. Let  $n=u[s']$, with $n$ numeral, and $m=s'_{i}(n)$. Let $\Phic_i=\Phic_{A(x,y)}$. If $A\notin\FSet$, then trivially $(i,n)\in \dm{s'}$ by definition \ref{definition-soundupdates}. Therefore for all $s''\geq s'$, $\Phic_{i} u[s'']=s''_{i}(n)=m$ and we are done. Hence, we may assume $A\in\FSet$. There are two cases, and this is the only point of this proof in which we use $\EM$.\\

\begin{enumerate}\item
$A(n,m)$ is true. Therefore, for all $s''\geq s'$, $s''_i(n)=m$ because $(i,n)\in\dm{s'}$. Thus, for all $s''\geq s'$, $\Phic_{i} u[s'']=s''_{i}(n)=m$, which is the thesis. \\

 \item
 $A (n,m)$ is false. If there exists $l$ such that  $A(n,l)$ is true, then let

$$s'':=\lambda x^{\Nat}\lambda y^{\Nat}.\, \ifthen{x=i\land_{\Bool} y={{n}}}{m}{s'_{x}(y)}$$

Then, for all $s'''\geq s''$, $s'''_{i}(n)=l$ because $(i, l)\in\dm{s''}$. Thus, for all $s'''\geq s''$, $\Phic_{i} u[s''']=s'''_{i}(n)=l$, which is the thesis.
If there is no $l$ such that  $A(n,l)$ is true, then trivially $(i,n)\in\dm{s'}$. Thus for all $s''\geq s'$, $\Phic_{i} u[s'']=s''_i(n)=m$ and we are done.
\end{enumerate}

\end{enumerate}

}

According to the Definition~\ref{definition-truthstate} of the truth value $\truthst{A}{s}$ of a formula $A$ in a state $s$, when we compute $\truthst{A}{s}$ we need only a finite number of Skolem function values, one for each quantifier of $A$. Thus, we can show with the excluded middle that for every state $s$ there exists a state $s'\geq s$ such that when we evaluate $A$ in the state $s'$ we obtain the real truth value of $A$.

{
\begin{proposition}\label{proposition-computabletruth}
Let $A(\vec{x})$ be any arithmetical formula and $\vec{n}$ be numerals. For every state $s$, there exists a state $s'\geq s$ such that $\truthst{A}{s'}(\vec{n})=\True$ if and only if $A(\vec{n})$ is true.
\end{proposition}

\textit{Proof}. We prove the thesis by induction on $A$.  Let $s$ be any state.  The cases in which $A$ is atomic or $A=B\lor C, B\land C, B\rightarrow C$ are trivial. Let us consider those in which $A$ starts with a quantifier.
\begin{itemize}

\item $A(\vec{n})=\exists y^{\Nat} B(\vec{n}, y)$. By the excluded middle, we extend $s$ to a state $s'\geq s$ such that  $m=s'_{B}\code{\vec{n}}$ implies that $$\exists y^{\Nat} B(\vec{n}, y)\rightarrow B(\vec{n}, m)$$

By induction hypothesis, there exists a state $s''\geq s'$ such that $B(\vec{n},m)$ is true if and only if

\[ \truthst{B}{s''}(\vec{n},m)=\truth{B}(\vec{n}, m)[s'']=\True\]

Assuming $\Phic_{i}=\Phic_{B}$, since $(i,\code{\vec{n}})\in \dm{s'}$, we have $s''_{B}\code{\vec{n}}=s'_{B}\code{\vec{n}}$. Since $$\truthst{A}{s''}(\vec{n})=\truth{B}(\vec{n}, \Phic_{B}\code{\vec{n}})[s'']= \truth{B}(\vec{n}, m)[s'']$$
and  $A(\vec{n})$ is equivalent to $B(\vec{n}, m)$, we obtain the thesis.
\vspace{2.5ex}

\item  $A(\vec{n})=\forall y^{\Nat} B(\vec{n}, y)$.  By the excluded middle, we extend $s$ to a state $s'\geq s$ such that  $m=s'_{B^{\bot}}\code{\vec{n}}$ implies that $$\exists y^{\Nat} B^{\bot}(\vec{n}, y)\rightarrow B^{\bot}(\vec{n}, m)$$

By induction hypothesis, there exists a state $s''\geq s'$ such that $B^{\bot}(\vec{n},m)$ is true if and only if

\[ \truthst{(B^{\bot})}{s''}(\vec{n},m)=\truth{B^{\bot}}[s''](\vec{n}, m)=\True\]

Assuming $\Phic_{i}=\Phic_{B^{\bot}}$, since $(i,\code{\vec{n}})\in \dm{s'}$, we have $s''_{B^{\bot}}\code{\vec{n}}=s'_{B^{\bot}}\code{\vec{n}}$. Since  $$\truthst{A}{s''}(\vec{n})=\truth{B^{\bot}}(\vec{n}, \Phic_{B^{\bot}}\code{\vec{n}})[s'']= \truth{B^{\bot}}(\vec{n}, m)[s'']$$ we obtain the thesis.

\end{itemize}

Now we prove a special case of the statement that the realizability of a formula implies the formula itself. Namely, we show that $t$ realizes $\bot$  implies $\bot$. The idea, as we have explained before, is to find a state $s$ which contains all the information needed to evaluate $t$.

\begin{theorem}[Consistency of Interactive Realizability]\label{theorem-consistency}
For every closed term $t$ of $\SystemTClass$, $ t \notreal \bot$. In particular, for every state $s$, there exists a state $s'\geq s$ such that $t\notreal_{s'} \bot$.
\end{theorem}

\textit{Proof}.
Suppose, for the sake of contradiction, that there exists a term $t$ such that $t\Vvdash \bot$. Let $s$ be any state. Since $t: \Update$, by theorem \ref{theorem-computability} we have $t\in \comp{\Update}$ and therefore there exists a state $r\geq s$ such that  $t\tdef{r}$. Let  $t[r]=\upconst{U}$ for some update $U$. Since $t\Vvdash_{r} \bot$, $U$ is non-empty: let $(i, \vec{n}, m)\in U$. By application of theorem \ref{theorem-computability}, if $\Phic_i=\Phic_{A}$, there exists a state $q\geq r$ such that
$\truth{A}(\vec{n}, m)\tdef{q}$. By definition,
 \[\truthst{A}{q}(\vec{n}, m)=\truth{A}(\vec{n}, m)[q]=b\]
for some boolean $b$. Since $t\Vvdash_{q} \bot$ and  $t[q]=\upconst{U}$ (because $t\tdef{r}$ and $q\geq r$), we obtain  by definition of realizability that $b=\True$. Let $q_{i}\code{\vec{n}}=l$. We have two possibilities:
\begin{enumerate}
\item $A(\vec{n}, l)$ is false. We define the state
\[s':=\lambda x^{\Nat}\lambda y^{\Nat}.\, \ifthen{x=i\land_{\Bool} y=\code{\vec{n}}}{m}{q_{x}(y)}\]
Then, $s'\geq q$, for $A(\vec{n}, l)$ is false. Moreover, since $\truth{A}(\vec{n}, m)\tdef{q}$, for all $q'\geq q$, $\truth{A}(\vec{n}, m)[q']=b$; by Proposition \ref{proposition-computabletruth}, there exists $q'\geq q$, such that $\truth{A}(\vec{n}, m)[q']=\True$ if and only if $A(\vec{n},m)$ is true. Since  $\truth{A}(\vec{n}, m)[q']=b=\True$, we have that $A(\vec{n},m)$ is true.
 By assumption on $t$, we have $t\Vvdash_{s'}\bot$ and $t[s']=\upconst{U}$, because $s'\geq r$. Since $s'_{i}\code{\vec{n}}=m$, by definition of $t\Vvdash_{s'} \bot$ we would have both $\truthst{A}{s'}(\vec{n},m)=\False$ and $\truthst{A}{s'}(\vec{n},m)=\True$, which is a contradiction. \\

\item $A(\vec{n}, l)$ is true. By Proposition \ref{proposition-computabletruth}, there is a state $s'\geq q$ such that $\truthst{A}{s'}(\vec{n},l)=\True$. By assumption on $t$, we have $t\Vvdash_{s'}\bot$ and $t[s']=\upconst{U}$. But $q_{i}\code{\vec{n}}=l$, $A(\vec{n}, l)$ is true and $s'\geq q$; therefore $(i,\vec{n})\in\dm{q}$ and  hence $s'_{i}\code{\vec{n}}=l$.  By definition of $t\Vvdash_{s'} \bot$, we would have  $\truthst{A}{s'}(\vec{n},l)=\False$ and $\truthst{A}{s'}(\vec{n},m)=\True$, which is in contradiction with  $\truthst{A}{s'}(\vec{n},l)=\True$.
\end{enumerate}

Finally, we are in a position to prove in $\HAw+\EM$ that the realizability of a formula $A$ implies its truth. For simplicity we assume $A$ is a  $\rightarrow$-free, but the result holds also in the general case.

\begin{theorem}[Soundness of Realizability]\label{theorem:RealTruth}
 Let $A$ be any $\rightarrow$-free arithmetical formula and suppose $t\Vvdash A$. Then $A$ is true.
\end{theorem}

\textit{Proof}.
We prove a stronger statement. Let $s$ be a state and suppose that for all $s'\geq s$, $t\Vvdash_{s'} A$. We prove by induction on $A$, that $A$ is true.

\begin{itemize}
\item $A=P$, with $P$ atomic. Suppose, by the way of contradiction, that $P$ is false. Then we have that for all $s'\geq s$, $t\Vvdash_{s'} \bot$, which is impossible by Theorem \ref{theorem-consistency}. \\

\vspace{-1ex}

\item $A=B\land C$. Then, for all $s'\geq s$, $t\Vvdash_{s'} A$ and $t\Vvdash_{s'} B$. By induction hypothesis $A$ and $B$ are true, and we obtain the thesis.\\

\vspace{-1ex}

\item $A=B \lor C$. By Theorem \ref{theorem-computability}, there exists a state $r\geq s$ such that $\proj_{0}t \tdef{r}$. Let $\proj_{0}t[r]=b$ with $b$ boolean, say $b=\True$. Then, by defintion,  for every $r'\geq r$,  $\proj_{0}t[r']=\True$ and therefore $t\Vvdash_{r'} A$. By induction hypothesis $A$ is true, and we obtain the thesis. \\

\vspace{-1ex}

\item $A=\forall x^{\Nat} B$. Let $n$ be any numeral. Then, for all $s'\geq s$, $tn\Vvdash_{s'} B(n)$. By induction hypothesis $B(n)$ is true. Therefore, $\forall x^{\Nat} B$ is true, and we obtain the thesis. \\

\vspace{-1ex}

\item $A=\exists x^{\Nat} B$. By Theorem \ref{theorem-computability}, there exists a state $r\geq s$ such that $\pi_{0}t \tdef{r}$. Let $\pi_{0}t[r]=n$ with $n$ numeral. Then, by definition,  for every $r'\geq r$,  $\pi_{0}t[r']=n$ and therefore $t\Vvdash_{r'} B(n)$. By induction hypothesis $B(n)$ is true, and we obtain the thesis.

\end{itemize}

Since all the proofs given in this section are formalizable in $\HAw+\EM$ (see Section~\ref{section-formalization}), we are able to prove the conservativity of $\HA^{\omega}+\EM+\SK$ over $\HA^{\omega}+\EM$ for arithmetical formulas.

\begin{theorem}[Conservativity of $\HA^{\omega}+\EM+\SK$ over $\HA^{\omega}+\EM$] Let $A$ be a closed arithmetical formula, and suppose $$\HA^{\omega}+\EM+\SK\vdash A$$ Then:
 \begin{equation}~\label{ris1}\HA^{\omega}+\EM\vdash A\end{equation}
   \begin{equation}~\label{ris2}\HA+\EM\vdash A \end{equation}

\end{theorem}

\textit{Proof}.\mbox{}\\

\vspace{-4.5ex}

 \begin{enumerate}\item We may assume that $A$ is $\redto$-free. Otherwise, $$\HA^{\omega}+\EM\vdash A\leftrightarrow B$$
with $B$ $\redto$-free and we consider $B$. Since $\FSet$ is arbitrary, we may assume that  all the subformulas of the instances of $\EM$ and $\SK$ used in the derivation belong to $\FSet$. By formalization of  the Adequacy Theorem~\ref{Realizability Theorem} in $\HA^{\omega}$ (see Section \ref{section-formalization}), we obtain that $\HA^{\omega}\vdash t\Vvdash A$ for some term $t$ of $\SystemTClass$. By formalization of the proof of Theorem~\ref{theorem:RealTruth} in $\HAw+\EM$, we obtain that $\HA^{\omega}+\EM\vdash (t\Vvdash A) \redto A$. We conclude $\HA^{\omega}+\EM\vdash A$. \\

\vspace{-2ex}

\item There are at least two ways to obtain the thesis. On one hand, we may use~(\ref{ris1}) and the standard result about the conservativity of $\HA^{\omega} + \EM$ over $\HA+\EM$ for arithmetical formulas (see for example Troesltra \cite{TroelstraMetamathematical}). On the other hand, we may code directly terms of system $\SystemTClass$ into natural numbers and then express the proofs of point 1) in $\HA+\EM$ (see Section \ref{section-formalization}).
\end{enumerate}

}

\section{Formalization of the Proofs in $\PA$ and in $\HA^{\omega}+\EM$}\label{section-formalization}

\vspace{-1.5ex}

In this section we explain how to formalize in $\PA$ and $\HAw+\EM$ the proof of the Adequacy Theorem~\ref{Realizability Theorem} of Section~\ref{section-ALearningBasedRealizability} and  the proofs of the Computability Theorem~\ref{theorem-computability} and the Soundness Theorem~\ref{theorem:RealTruth} of Section \ref{section-Conservativity}.
We start with the case of $\PA$.

\vspace{-1.8ex}

\subsection{Formalization in $\PA$}\label{sec:formPA}

\vspace{-1.8ex}

One can routinely code in $\PA$ all the concepts we have so far used. As in Tait \cite{Tait}, one may code the terms of $\SystemTClass$ with natural numbers and successively the definition of the realizability and computability predicates with arithmetical formulas. Since neither set-theoretic concepts nor Skolem axioms are employed in any of the given proofs, everything can be coded in $\PA$.

\subsection{Formalization in $\HAw+\EM$}\label{sec:formHAwEM}

\vspace{-1.8ex}

Instead of coding everything into natural numbers, which is of limited practical interest, it is more satisfying to formalize our proofs directly in $\HAw+\EM$. There is no serious obstacle to this end, except for a small formalization issue: the notion $t[s]$ of evaluation of a term $t$ of $\SystemTClass$ in a state $s$, which we have heavily used in the definitions of the realizability and computability predicates, is not directly representable in $\HAw+\EM$.
To begin with, terms of $\SystemTClass$ may contain some constant $\Phic\in\skolemcon$ which does not belong to the language of $\HAw$. This problem is easily solved by considering terms of the form $t[s]$ with $s$ state variable. However, in the definition of Interactive realizability for implication and in the statement of the Computability Theorem one needs to define formulas $x\Vvdash A$ and $x\in \comp{\Nat}$, where $x$ is a variable. In these definitions it is necessary to speak about the substitution of an actual state $s$ in the body of a variable $x$, which is impossible  in $\HAw$ (remember that $x$ represents a term $t[s]$ of $\SystemT$). This last issue is overcome quite easily by considering in place of a term $t:\tau$ in $\SystemTClass$ the term $\lambda s^{\State}. t[s]:\State\redto \tau$, where $\State:=\Nat^{2}\redto\Nat$ is the type of states. In this way, one makes explicit the functional dependence of $t$ from the state $s$ and transforms $t$ into an object having a semantical denotation. It is however necessary to slightly adapt the definitions of realizability and computability, which is what we are going to do.

First, we give an alternative definition of Interactive realizability, which is shown in~\cite{ABF} to be equivalent to
 Kreisel's modified realizability for $\HA^{\omega}$ applied to some  Friedman translation of formulas.
We denote with $\Language$ the restriction of the language $\LanguageClass$ to  the formulas not containing  any Skolem function  constant $\Phic\in\skolemcon$.

\begin{definition}[Alternative Definition of Interactive Realizability]\label{definition-mrsf}
 Assume $s:\State$ is a closed term of $\SystemT$, $t$ is a closed term of $\SystemT$, $D \in \Language $ is a closed formula of $\Language$, and $t:|D|$. We  define by induction on $D$ the relation $t\ \mrsf_s\ D$:
\begin{enumerate}
 \item $t\ \mrsf_s\  Q$ if and only if $t=\upconst{U}$ implies:
 \begin{itemize}
\item  for every  $(i, \vec{n}, m)\in U$,  $\Phic_i=\Phic_A$ for some $A\in\FSet$, and $\truthst{A}{s}(\vec{n}, s_i \code{\vec{n}})=\False$ and $\truthst{A}{s}(\vec{n},m)=\True$.
\item $U  = \varnothing$  implies
$Q={\True}$\\
\end{itemize}

\vspace{-1ex}

\item
$t\ \mrsf_s\ {A\wedge B}$ if and only if $\pi_0t \ \mrsf_s\ {A}$ and $\pi_1t\ \mrsf_s\ {B}$\\

\vspace{-1ex}

\item
$t\ \mrsf_{s} \ {A\vee B}$  if and only if either $\proj_0t={\True}$ and $\proj_1t\ \mrs\ A$, or $\pi_0t={\False}$ and $\proj_1t\ \mrs\ B$ \\

\vspace{-1ex}

\item
$t\ \mrsf_s\  {A\rightarrow B}$ if and only if for all $u$, if $u\ \mrsf_s\ {A}$,
then $tu\ \mrsf_s\ {B}$\\

\vspace{-1ex}

\item
$t\ \mrsf_{s}\ {\forall x^{\tau} A}$ if and only if for all closed terms $u:\tau$ of $\SystemT$,
$tu\ \mrsf_{s}\ A[{u}/x]$\\
\item

\vspace{-1ex}

$t\ \mrsf_{s}\ \exists x^{\tau} A$ if and only for some closed term $u:\tau$ of $\SystemT$, $\pi_0t= u$  and $\pi_1t\ \mrsf_{s}\ A[{u}/x]$

\end{enumerate}

\end{definition}

One can prove straightforwardly, as in \cite{ABF}, that our first Definition~\ref{definition-IndexedRealizabilityAndRealizability} of Interactive realizability is equivalent to this alternative one.

\begin{theorem}[Characterization of Interactive Realizability]\label{theorem-characterization}
Let $t\in \SystemTClass$ and $s$ be a state.
Then, for every  $B\in \LanguageClass$
\[t\,\Vvdash_s\, B\iff t[s]\, \mrsf_{s}\, B[s]\]
\end{theorem}

Theorem~\ref{theorem-characterization} allows us to replace in  our conservativity proof the expression $t\Vvdash A$ with the expression $\forall s^{\State}.\, t[s]\, \mrsf_{s}\, A[s]$, which is a formula of $\HAw$. Moreover, the Adequacy Theorem for $\Vdash$ is formalizable in $\HAw$, since it is a special case of the Adequacy Theorem for modified realizability, which is formalizable in that system (see \cite{TroelstraHandBook}).

\newcommand{\app}[2]{{#1}\cdot{#2}}
Secondly, we adapt the notion of computability to terms of type $\State\redto\tau$.  For every pair of terms $t,u\in\SystemT$ respectively of type $\State\redto(\sigma\redto\tau)$ and $\State\redto\sigma$, we define the following notion of application:

\vspace{-2ex}

$$\app{t}{u}:= \lambda s^{\State}. ts(us)$$

\newcommand{\pr}{\uppi}

For every term $t\in\SystemT$ of type $\State\redto(\tau_0\times\tau_1)$ and $i\in\{0,1\}$, we define the following notion of projection:
$$\pr_i t:=  \lambda s^{\State}. \pi_{i}  ts$$

\newcommand{\starc}[1]{#1^{*}}
 Finally, for every constant term $\mathsf{c}\notin\skolemcon$,  we define $\starc{\mathsf{c}}:=\lambda s^{\State} \mathsf{c}$.
We now adapt Definition~\ref{definition-defstate} and Definition~\ref{definition-Computability}. Since there is no possibility of confusion, we maintain the same notations of Section \ref{section-Conservativity} but with the new specified meaning.

\begin{definition}[Definition of a term in a state $s$]~\label{definition-defstate2} For every state $s$ and  term $t:\State\redto\tau$ of $\SystemT$ with $\tau$ atomic type, we define  $t\tdef{s}$ (and we say ``$t$ is defined in $s$'') as the statement: for all states $s'\geq s$, $ts'=ts$.
\end{definition}

\begin{definition}[Computable terms]\label{definition-Computability2}\mbox{}

For every type $\tau$ of $\SystemT$, we define a set of closed terms of $\SystemT$ of type $\State\redto\tau$ as follows:

\begin{itemize}
\item $\comp{\Nat}$=$\{t:\State\redto\Nat\;|$  \mbox{for all states  $s$ there  is a state  $s'\geq s$ such that $t\tdef{s'}$} $\}$ \\

\vspace{-1.5ex}

\item $\comp{\Bool}$=$\{t:\State\redto\Bool\;|$  \mbox{for all states  $s$ there  is a state  $s'\geq s$ such that  $t\tdef{s'}$} $\}$ \\

\vspace{-1.5ex}

\item $\comp{\Update}$=$\{t:\State\redto\Update\;|$  \mbox{for all states  $s$ there  is a state  $s'\geq s$ such that  $t\tdef{s'}$} $\}$ \\

\vspace{-1.5ex}

\item $\comp{\tau\redto \sigma}$=$\{t\;|\; \forall u\in\comp{\tau}\hspace{2pt} \;\app{t}{u}\in\comp{\sigma}\} $ \\

\vspace{-1.5ex}

\item $\comp{\tau\times \sigma}$=$\{t\;|\; \pr_{0}t\in\comp{\tau} \mbox{and } \pr_{1}t\in\comp{\sigma}\} $

\end{itemize}
\end{definition}

The proofs of  Lemma~\ref{lemma:saturation} and of the Computability Theorem  can be easily adapted (for details, see the full version of this paper \cite{RattoRatto}).


\begin{lemma}\label{lemma:saturation2}
 For every term $t:\State\redto\rho$ of $\SystemT$, if for every state $s$ there exists a state $s'\geq s$ and $u\in\comp{\rho}$ such that for all states $s''\geq s'$, $ts''=us''$, then $t\in\comp{\rho}$.
\end{lemma}

\begin{theorem}[Computability Theorem]\label{theorem-computability2}\mbox{}

\noindent
Let $v:\tau$ be a term of $\SystemTClass$ and suppose that all the free variables of $v$ are among $x_1^{\sigma_1},\ldots, x_{n}^{\sigma_n}$.
If $t_{1}\in\comp{\sigma_1}, \ldots, t_{n}\in\comp{\sigma_n}$, then $\lambda s^{\State}.v[s][t_{1}s/x^{\sigma_1}_{1},\ldots, t_{n}s/x^{\sigma_n}_{n}]\in\comp{\tau}$.
\end{theorem}

The proofs of Proposition \ref{proposition-computabletruth} and  Theorem~\ref{theorem-consistency} remain exactly the same, while the proof of Theorem~\ref{theorem:RealTruth} can be straightforwardly adapted. In particular, in the base case of the induction one needs to prove that a term $t$, possibly with free variables of type $\Nat$, is computable. This follows from Theorem~\ref{theorem-computability2} and the fact that it is possible to prove by induction the statement  $\forall x^{\Nat}. \;\lambda s^{\State} x\in\comp{\Nat}$.


\comment{

}


\end{document}